%% file: SBpaper.tex
\documentclass[submission,copyright,creativecommons]{eptcs}
\providecommand{\event}{SOS 2007} 

\usepackage{iftex}
\usepackage{mathtools}
\usepackage{proof}

\ifpdf
  \usepackage{underscore}         
  \usepackage[T1]{fontenc}        
\else
  \usepackage{breakurl}           
\fi
\usepackage{amssymb}
\usepackage{amsmath}
\usepackage{amsthm}
\usepackage{stmaryrd}
\theoremstyle{definition}
\newtheorem{definition}{Definition}[section]
\newtheorem{theorem}[definition]{Theorem}
\newtheorem{lemma}[definition]{Lemma}

\newtheorem{example}[definition]{Example}

\input{prooftree}

\newcommand{\weg}[1]{}

\def\phi{\varphi}
\def\squareforqed{\hbox{\rlap{$\sqcap$}$\sqcup$}}
\def\qed{\ifmmode\squareforqed\else{\unskip\nobreak\hfil
\penalty50\hskip1em\null\nobreak\hfil\squareforqed
\parfillskip=0pt\finalhyphendemerits=0\endgraf}\fi}

\newcommand{\inty}[1]{\llparenthesis{#1}\rrparenthesis}
\newcommand{\intyr}[1]{\inty{#1}_{\rho}}
\newcommand{\intyrc}[2]{\inty{#1}_{#2}}

\newcommand{\inte}[1]{\llbracket{#1}\rrbracket}
\newcommand{\intexr}[1]{\inte{#1}_{\xi\rho}}
\newcommand{\intexrc}[2]{\inte{#1}_{#2}}
\newcommand{\inteAr}[1]{\inte{#1}^{\cA}_{\rho}}
\newcommand{\inteArx}[2]{\inte{#1}^{\cA}_{\rho(x:=#2)}}
\newcommand{\cV}[1]{{\cal V}({#1})}
\newcommand{\cVxr}[1]{\cV{#1}_{\xi\rho}}
\newcommand{\cVxrc}[2]{\cV{#1}_{#2}}

\newcommand{\ca}{{\rm CA}}
\newcommand{\weca}{{\rm WECA}}

\newcommand{\Models}{\mathrel{{|}\!{\models}}}
\newcommand{\nat}{{\sf nat}}
\newcommand{\suc}{{\sf succ}}
\newcommand{\zero}{{\sf O}}
\newcommand{\bool}{{\sf bool}}
\newcommand{\sI}{{\sf I}}
\newcommand{\true}{{\sf true}}
\newcommand{\false}{{\sf false}}
\newcommand{\usuc}{\underline{{\sf succ}}}
\newcommand{\uzero}{\underline{{\sf O}}}
\newcommand{\ind}{{\sf ind}}
\newcommand{\Ind}{{\sf Ind}}
\newcommand{\Lim}{{\sf Lim}}

\newcommand{\lampt}{\lambda P2}
\newcommand{\lamt}{\lambda 2}
\newcommand{\CC}{{\rm CC}}

\newcommand{\Var}{{\rm Var}}
\newcommand{\FV}{{\rm FV}}
\newcommand{\Vo}{\Var^{\star}}
\newcommand{\Vc}{\Var^{{\sf Kind}}}
\newcommand{\Boxx}{{\sf Kind}}

\newcommand{\llam}{{\boldmath \lambda}}
\newcommand{\lamstar}{\lambda^*}

\newcommand{\T}{{\sf T}}

\newcommand{\cT}{{\cal T}}

\newcommand{\cN}{{\cal N}}
\newcommand{\cP}{{\cal P}}
\newcommand{\cA}{{\cal A}}
\newcommand{\cS}{{\cal M}}

\newcommand{\bA}{{\bf A}}
\newcommand{\bk}{{\bf k}}
\newcommand{\bK}{{\bf K}}
\newcommand{\bKs}{{\bf K}_*}
\newcommand{\bI}{{\bf I}}
\newcommand{\bs}{{\bf s}}

\newcommand{\bone}{{\bf 1}}
\newcommand{\bLambda}{{\Lambda}}
\newcommand{\bLambdaid}{{\Lambda}{{\sf id}}}
\newcommand{\gen}[2]{\cP^{#1}(#2)}
\newcommand{\genC}{\gen{\bLambda}{C}}

\newcommand{\genidrefl}{\gen{\bLambdaid}{\{\refl\}}}

\newcommand{\Prodk}{\prod}
\newcommand{\Prodt}{\Pi}

\newcommand{\oftype}{:}

\newcommand{\red}{\longrightarrow}

\newcommand{\redJ}{\red_J}
\newcommand{\redb}{\red_{\beta}}
\newcommand{\redrefl}{\red_{\refl}}
\newcommand{\redsig}{\red_{\Sigma}}
\newcommand{\arr}{\rightarrow}

\newcommand{\IN}{{\rm I\kern-.23em N}}

\usepackage{tikz}
\usetikzlibrary{automata, positioning, arrows, shapes, decorations.pathmorphing,calc}
\tikzstyle{curvy}=[dashed,decorate,yshift=0.3cm,decoration={snake,amplitude=3mm,segment length=8mm,post length=1mm}]
\tikzstyle{block}=[draw, rectangle]
\tikzset{
    ->,
    >=stealth,
    node distance=3cm,
    every state/.style={line width=0.8pt},
    initial text=$ $,
    minimum size=2em,
    inner sep=1
}
\usepackage{tikz-cd}
\newcommand{\tup}[1]{\langle#1\rangle}
\newcommand{\pr}[1]{\pi_{#1}}

\newcommand{\refl}[0]{{\sf refl}}
\newcommand{\Jcomb}[0]{{\sf J}}
\newcommand{\U}[0]{\star}
\newcommand{\betar}[0]{=_{\beta}}

\newcommand{\funext}{{\sf FunExt}}
\newcommand{\UIP}{{\sf UIP}}

\newcommand{\bisim}{\sim}

\newcommand{\pack}{{\sf pack}}
\newcommand{\RecExists}{{\sf rec}_\exists}

\newcommand{\Stream}{{\sf Stream}}
\newcommand{\StreamB}{\Stream_{\bool}}

\newcommand{\hd}{{\sf hd}}
\newcommand{\tl}{{\sf tl}}
\newcommand{\CoRecStream}{{\sf corec}_{{\sf s}}}

\newcommand{\Resp}{{\sf Resp}}

\newcommand{\clsr}{{\sf cls}^*}

\newcommand{\IndQuot}{{\sf IndQuot}}

\newcommand{\quot}{{\sf quot}}
\newcommand{\LimQuot}{{\sf LimQuot}}
\newcommand{\RecQuot}{{\sf rec}_{{\sf q}}}
\newcommand{\cls}{{\sf cls}}

\title{Non-Derivability Results in Polymorphic Dependent Type Theory}
\author{Herman Geuvers
  \institute{Radboud University\\ Nijmegen, The Netherlands\\
    \and
    Eindhoven University of Technology\\
  The Netherlands}
\institute{Institute for Computing and Information Science\\
Faculty of Science}
\email{herman@cs.ru.nl}
}

\newcommand{\titlerunning}{Non-Derivability Results in $\lampt$}
\newcommand{\authorrunning}{H. Geuvers}

\hypersetup{
  bookmarksnumbered,
  pdftitle    = {\titlerunning},
  pdfauthor   = {\authorrunning},
  pdfsubject  = {EPTCS},               
  pdfkeywords = {Type theory, data types, formulas-as-types, realizability semantics} 
}

\begin{document}
\maketitle

\begin{abstract}
In the pure Calculus of Constructions (\CC)  one can
define data types and function over these, and there
is a powerful higher order logic to reason over these functions and data types \cite{CoquandHuet88,Berardi1993}. This is due to
the combination of impredicativity and dependent types, and most of
these features can already be observed in polymorphic (second order) dependent type
theory $\lampt$. The impredicative encoding of data
types (in $\lampt$ or \CC) is powerful but not fully satisfactory: for example, the
induction principle is not provable. As a matter of fact, it can be
shown that induction is not provable for whatever possible representation of
data types\cite{NoInduction}. In \cite{encodings}, Awodey, Frey and
Speight show that in an extension of $\lampt$ with a $\Sigma$-types,
identity types with uniqueness of identity proofs and function
extensionality, it is possible to define data types for which the
induction principle is provable. In \cite{BronsveldGvdW} it has been
shown that in this extension of $\lampt$, also quotient types can be
defined with the proper induction principle, and, using quotient types, coinductive types can
be defined with the proper coinduction principle.

This leaves various questions open: Are quotient types with induction
principle not definable in the original $\lampt$? And how about
coinduction types, is it impossible to get a strong coinduction principle in $\lampt$?  Looking at it from the other side: which of the
extensions used in \cite{encodings} are really needed to make induction and coinduction work?

In this paper, we contribute partial answers to these questions:
parametric quotient types are not definable in $\lampt$ and the
well-known definable stream type does not have a coinduction
principle. For the latter question we show that, if we just extend
$\lampt$ with $\Sigma$-types and identity types with uniqueness of
identity proofs, we still cannot prove an induction principle for the
natural numbers. So function extensionality is crucial in making
induction provable.

We show these results by studying
models of $\lampt$ where the types representing these principles are
empty, so these models act as counter models to the derivability of
the principles. Most of the well-studied models for polymorphic type
theory are parametric or in any case quite extensional
\cite{BainbridgeFSS1990,PlotkinAbadi1993,Hasegawa1994}, so in these
models, induction and coinduction just hold. An interesting exception
is work by Berardi and co-authors
\cite{BerardiBerline2002,BarbaneraBerardi2003,BerardiBerline2004}, who
describe non-parametric models of (non-dependent) polymorphic type
theory. The models we study here are of a more syntactic nature and
have been described and studied in
\cite{NoInduction,Geuvers1996}. They can be seen as a collection of
term models where specific choices can be made to construct counter
models.\footnote{Acknowledgments: Thanks to the referees for their useful and constructive comments.}
\end{abstract}

\section{Introduction}
This paper is dedicated to Stefano Berardi, to celebrate his 40th\footnote{in the appropriate numeral system}
birthday, and maybe even more to celebrate the numerous scientific
contributions he has made to logic and type theory. I met Stefano for
the first time in the fall of 1988, in Pisa where he came to visit
from Torino. I was spending the fall semester in Pisa, and Stefano
stayed with us in Tirrenia a number of nights. I was immediately
impressed by his broad knowledge, his deep insights and his original
ideas and I have been ever since. On top of that, Stefano is a
very nice person who generously shares his knowledge, a true
scientist. Thank you very much and congratulations Stefano!

The content of the present paper goes back to research that Stefano
was doing for his PhD.\ thesis and soon after that (e.g.\ see
\cite{Berardi1993}), looking into representing data types, functions
and reasoning principles over them in the Calculus of Constructions.
It was noticed, and it was kind of folklore at the time, that one needs
to add induction principles because they are not derivable. Later I
formally proved that indeed, induction is not derivable, for whatever
smart encoding of the data types you use \cite{NoInduction}. The
present paper continues this thread by showing that various other
reasoning principles are not derivable. To simplify the work we limit
ourselves to second order dependent type theory as opposed to the full
Calculus of Constructions, but we are confident the results can be
extended straightforwardly.

In second order dependent type theory, $\lampt$, we can encode all kinds
of inductive data types, like the types of natural numbers, lists,
trees etcetera (see \cite{GirardTaylorLafont93} for a general exposition), which yields
e.g.\ the well-known polymorphic Church numerals as interpretation of
the natural numbers. This encoding already works for non-dependent
second order type theory (the well-known polymorphic
$\lambda$-calculus $\lamt$), but dependent types give the extra
advantage that we can also state the {\em induction principle\/} for
the inductive data types. For example, if $\nat$ is the well-known type of
polymorphic Church numerals with zero $\zero$ and successor function
$\suc$\footnote{As a reminder: $\nat := \Pi \alpha:\U.\ \alpha\to(\alpha\to\alpha)\to\alpha$, $\suc := \lambda y:\nat.\ \lambda \alpha:\U.\ \lambda x: \alpha.\ \lambda f: \alpha\to\alpha.\ f\ (y\ \alpha\ \ x\ f)$ and $\zero := \lambda \alpha:\U.\ \lambda x: \alpha.\ \lambda f: \alpha\to\alpha.\ x$}, then the induction principle is represented by the type
$\ind$ defined as
$$\ind := \Pi P: \nat\arr\star.\ P\ \zero \arr (\Pi y\oftype \nat .\ P\
y \arr P\ (\suc\ y))\arr \Pi x\oftype \nat.\ P\ x.$$
Here, $\star$ denotes the `kind' (universe) of all types, which
captures both the sets ($\nat : \star$) and the propositions ($\ind :
\star$). 
The induction principle for $\nat$ is said to be {\em derivable\/} in
$\lampt$ if there is a closed term of type $\ind$.

It is well-known that the induction principle for $\nat$ is not
derivable in $\lampt$. A proof specifically for $\nat$, is given in
\cite{Streicher92}, using realizability semantics, and also in
\cite{Rummelhoff2004} using PER models.  But the situation is more
general, as we have shown in \cite{NoInduction}: If $N$ is a type with
$0:N$ and $s:N\arr N$, then there cannot be a (closed) term of type
$\Pi P: N\arr\star.\ P\ \zero \arr (\Pi y\oftype N .\ P\
y \arr P\ (\suc\ y))\arr \Pi x\oftype N.\ P\ x$. So there are no smart ways to encode
the natural numbers in $\lampt$ as an initial algebra. An obvious
smart encoding would be by relativizing to the ``inductive natural
numbers'', internalizing the induction principle into the notion of a
natural number. Then one would define $N$ as follows.
\begin{eqnarray*}
  \Ind\  x &:=& \Pi P: \nat\arr\star.\ P\ \zero \arr (\Pi y\oftype \nat .\ P\
y \arr P\ (\suc\ y))\arr P\ x,\\
N &:=& \exists
x\oftype\nat.\ \Ind\  x.
\end{eqnarray*}
Here, $\exists$ is defined in the well-known second order
way: $\exists x\oftype\sigma.\tau := \Pi \alpha \oftype\star.\ (\Pi
x\oftype\sigma.\ \tau\arr\alpha)\arr\alpha$. By using the definable
$\exists$-elim and $\exists$-intro rules, it is now easy to define
$\uzero$, $\usuc$ for this encoding:
\begin{eqnarray*}
  \uzero&:=& \lambda\alpha:\star.\ \lambda h\oftype \Pi
  x\oftype\nat.\ \Ind\  x\arr\alpha.\ h\ \zero\ q_{\zero},\\
  \usuc&:=& \lambda n\oftype N.\ n\  N\ \big(\lambda x\oftype \nat.\ \lambda
  p\oftype \Ind \  x.\\
	&& \lambda\alpha\oftype\star.\ \lambda h\oftype \Pi
  y\oftype\nat.\ \Ind\  y\arr\alpha).\ h\ (\suc\  x)\ (q_{\suc}\ x\ p)\big),
\end{eqnarray*} 
where $q_{\zero}$ and $q_{\suc}$ are terms such that $q_{\zero}:
\Ind\ \zero$ and $q_{\suc}: \Pi x\oftype\nat.\ \Ind\ x\arr
\Ind\ (\suc\ x)$.  But again, also for these natural numbers,
induction is not provable in $\lampt$, as \cite{NoInduction} proves.

In \cite{encodings}, Awodey, Frey and Speight have shown that in a
moderate extension of $\lampt$, with concepts inspired from Homotopy
Type Theory, induction is provable for natural numbers, taking the
``relativized'' definition above, with strong $\Sigma$-types instead
of the definable (weaker) $\exists$-types. In
fact, \cite{encodings} relativizes with a different predicate called
$\Lim$ which is a direct translation of the concept of initial
algebra: $\Lim$ singles out the terms of type $N$ that turn $N$ into a
type satisfying the initial algebra property. In \cite{BronsveldGvdW}
we have shown that one can relativize in various ways, by translating
the categorical initiality property.
The type theory of \cite{encodings} that induction can be proven in is the extension of $\lampt$ with
\begin{enumerate}
\item Identity types,
\item Uniqueness of Identity Proofs (\UIP),
\item Function extensionality (\funext),
\item $\Sigma$-types.
\end{enumerate}
In \cite{BronsveldGvdW} we have shown that in this extension of
$\lampt$, also quotient types (as initial algebras) can be defined with the proper induction
principle, and using those, coinductive types can be defined with the proper
coinduction principle.  This also raises some questions: can quotient
types not already be defined in $\lampt$ itself? And what about
coinductive types, are the definable coinductive type (e.g.\ streams)
not already terminal co-algebras? And if not, isn't there a `smart'
encoding of terminal co-algebras in $\lampt$? From the other side, one
may wonder which of the extensions of $\lampt$ are really needed to
obtain induction: is function extensionality really needed? Are
$\Sigma$-types really needed?

\subsection{Overview of main results}

In the present paper we give some partial answers to the questions, and we do
so mainly by syntactic analysis and by constructing {\em counter-models}: we construct a model of
$\lampt$ where there are no {\em parametric\/} quotient types, that is, quotient where the quotient operations are polymorphic. We also show a
model in which the well-known (polymorphically definable) stream data type is not a terminal co-algebra. Finally, we study the extension of $\lampt$ with $\Sigma$-types and identity types (as in \cite{encodings}) and we give a model of these where
\UIP\ holds, but there is no induction for natural numbers.

The models we construct are of a syntactical nature, variations on the
models that we considered in \cite{NoInduction} to show that induction is
not provable in ``vanilla $\lampt$''. In that respect, this work can
be seen in the style of Stefano Berardi, who has also studied more
syntactical models, e.g.\ models of polymorphic $\lambda$ calculus
that are not parametric
\cite{BerardiBerline2002,BarbaneraBerardi2003,BerardiBerline2004},
that are based on models of untyped $\lambda$ calculus. Our models are
similar in spirit, but based on arbitrary {\em weakly extensional combinatory
  algebras}, \weca. The prime example of a \weca\ is the set of (open)
untyped $\lambda$-terms modulo $\beta$ or $\beta\eta$, so we will use
that a lot.

\section{Second order dependent type theory}
Second order dependent type theory, $\lampt$, extends polymorphic $\lambda$-calculus ($\lamt$ or system $F$) with dependent types
and it was first
introduced as such in \cite{LongoMoggi91}. It is a subsystem of the
Calculus of Constructions (\cite{CoquandHuet88}, \cite{Coquand89}), where the
operations of forming type constructors are restricted to second order
ones. (So, there are type constructors of kind $\sigma
\arr \star$, but not of kind $(\sigma
\arr \star)\arr \star$.) It can also be seen as an extension of the
first order system $\lambda P$, where quantification over type
constructors has been added. For an extensive discussion on these
systems and their relations, we refer to \cite{Barendregt92} or
\cite{NederpeltGeuvers14}. Here we just define the system $\lampt$ and give some
initial motivation for it.

\begin{definition} The set of
{\em pseudo-terms},$\T$, of the type system $\lampt$ is given by 
$$ \T ::=
\star\ |\ \Boxx \ |\  \Var \ |\ (\Pi \Var \oftype \T .\ \T) \ |\ 
(\lambda \Var \oftype \T .\ \T) \ |\  \T\ \T ,$$
where $\Var$ is a countable set of variables.
On $\T$ we have the usual
notions of $\beta$-reduction,  $\red_{\beta}$. The
transitive reflexive symmetric closure of $\red_{\beta}$ is denoted by
$=_{\beta}$.

The deduction rules for $\lampt$ derive judgments of the forms $\Gamma \vdash M:T$, with $M$ and $T$
pseudo-terms and $\Gamma$ a {\em context}, a finite
sequence of declarations $v_1 \oftype T_1, \ldots ,v_n \oftype T_n$. Here, $v$ ranges over $\Var$, $s$,
$s_1$ and $s_2$ range over $\{\star,\Boxx\}$ and $M,N,T$ and $U$ range over
$\T$.

\begin{align*}
  \begin{aligned}
    \infer[\text{axiom}]{\star:\Boxx}{%
    }
    \qquad
           &
    \infer[\text{var}]{\Gamma, v: T \vdash M: U}{%
       \Gamma \vdash T : \star / \Boxx  
           &
      \Gamma \vdash M:U
    }
    \\\\
    \infer[\text{weak}]{\Gamma, v: T \vdash M: U}{%
      \Gamma, v: T \vdash M: U
      &
      \Gamma \vdash T : \star / \Boxx
    }
    \qquad
    &
    \infer[\text{conv$_{\beta}$, if }T=_{\beta}U]{\Gamma \vdash M: U}{%
      \Gamma\vdash M: T
           &
      \Gamma \vdash U : \star / \Boxx
    }
    \\\\    
    &
    \infer[\Pi\text{ if }(s_1,s_2) \neq (\Boxx,\Boxx)]{\Gamma \vdash \Pi x: U. T : s_2}{%
      \Gamma \vdash U : s_1
           &
      \Gamma, x: U \vdash T : s_2
    }
    \\\\
    \infer[\lambda]{\Gamma \vdash \lambda x: U. b : \Pi x: U. T}{%
      \Gamma \vdash \Pi x: U . T : s
      &
      \Gamma, x: U \vdash b : T
    }
    \qquad
           &
    \infer[\text{app}]{\Gamma \vdash f\ a : T[x:=a]}{%
      \Gamma \vdash f : \Pi x: U. T
           &
      \Gamma \vdash a: U
    }
  \end{aligned}
\end{align*}

In the rules (var) and (weak) it is always assumed that
the newly declared variable is fresh, that is, it has
not yet been declared in $\Gamma$. 
\end{definition}

For notational convenience, we split up
the set $\Var$ into a set $\Vo$, the {\em term variables}, and $\Vc$, the
{\em constructor variables}. Term variables are denoted by $x, y, z,
\ldots$ and constructor variables by $\alpha, \beta ,\ldots$. In the rules (var) and
(weak), we take the variable $v$ out of $\Vo$ if $s = \star$ and out of $\Vc$
if $s = \Boxx$.
The well-typed terms can be split into the following disjoint subsets:
\begin{itemize}
\item $\{\Boxx\}$, 
\item the set of {\em kinds}: expressions $A$ such that $\Gamma \vdash
A:\Boxx$ for some $\Gamma$; this includes $\star$.\\
In $\lampt$ all kinds are of the form $\Pi x_1\oftype \sigma_1 \ldots
\Pi x_n \oftype \sigma_n . \star$, with $\sigma_1 ,\ldots ,\sigma_n$ 
{\em types\/} and $x_1 ,\ldots ,x_n \in \Vo$.
\item the set of {\em constructors}: expressions of type a `kind', i.e.\
expressions $P$ such that $\Gamma \vdash P:A$ for some kind $A$; this
includes the {\em types}, expressions of type $\star$.\\
In $\lampt$ all constructors are of one of the following forms
\begin{itemize}
\item $\alpha \in \Vc$,
\item $P\ t$, with $P$ a constructor and $t$ an {\em term},
\item $\lambda x\oftype \sigma.\ P$, with $\sigma$ a type, $P$ a constructor, 
$x\in \Vo$,
\item $\Pi x\oftype \sigma.\ \tau$, with $\sigma$ and $\tau$ types, $x\in \Vo$,
\item $\Pi \alpha\oftype A.\ \tau$, with $A$ a kind, $\tau$ a type, 
$\alpha\in \Vc$.
\end{itemize}
\item the {\em terms}: expressions of type a `type', i.e.\ expressions $M$ such
that $\Gamma \vdash M:\sigma$ for some type $\sigma$.
In $\lampt$ all terms are of one of the following forms
\begin{itemize}
\item $x \in \Vo$,
\item $q\ t$, with $q$ and $t$ an terms,
\item $q\ P$, with $P$ a constructor and $q$ an term,
\item $\lambda x\oftype \sigma.\ t$, with $\sigma$ a type, $t$ an term, 
$x\in \Vo$,
\item $\lambda \alpha\oftype A.\ t$, with $A$ a kind, $t$ an term, 
$\alpha\in \Vc$.
\end{itemize}\end{itemize}
As a convention, we denote kinds by $A, B, C,\ldots$, types by
$\sigma,\tau,\ldots$, constructors by $P,Q, \ldots$ and terms by $t,
q, \ldots$.

\subsection{Extensions}
In \cite{encodings} and in \cite{BronsveldGvdW}, the system $\lampt$
is extended with various features to define a data type for natural
numbers that has induction and to define quotient types and
co-inductive types with a co-induction principle. We list these
extensions here, and we will consider some of these in the models of
$\lampt$ to see which extensions are needed to accomplish the data
types of \cite{encodings,BronsveldGvdW}.

\begin{definition}\label{def.extensions}
We consider the following rules as possible extensions of $\lampt$.\\
{\bf \underline{Identity Type}}
$$\begin{array}{c}
     \infer[\text{identity}]{\Gamma \vdash a =^{\sigma} b : \U}{%
      \Gamma \vdash \sigma : \U
           &
      \Gamma \vdash a : \sigma
           &
      \Gamma \vdash b : \sigma
    }
    \qquad
        \infer[\text{refl}]{\Gamma \vdash \refl : a =^{\sigma} a}{%
      \Gamma \vdash \sigma : \U
           &
      \Gamma \vdash a : \sigma
    }
\\\\
  \infer[\text{J}]{\Gamma \vdash \Jcomb (c,a,b,q) : \tau[x:=a, y:=b, p:=q]}{%
    \begin{gathered}
      \Gamma \vdash a : \sigma \qquad
      \Gamma \vdash b : \sigma \qquad
      \Gamma \vdash q : a =^{\sigma} b
      \\
      \Gamma, x: \sigma, y: \sigma, p: x =^{\sigma} y \vdash \tau : \U \qquad
      \Gamma \vdash c : \Pi z: \sigma. \tau[x,y:=z,p := \refl]
    \end{gathered}
  }
\end{array}$$
The rules for (identity) comes with an equality rule:
$$\Jcomb (c,a,a,\refl) \redrefl c\ a$$
For identity types, the following two principles may be considered, {\em functional extensionality}, \funext, and {\em uniqueness of identity proofs}, \UIP. (We omit the uppercase labels in the identity types, so we write $x = y$ for $x =^{\alpha} y$ etc.)  
\begin{align*}
  \funext \quad & : \quad \Pi f, g : (\Pi x : \sigma.\ \tau) .\ (\Pi x: \sigma.\ f\ x = g\ x ) \implies f = g \\
  \UIP \quad    & : \quad \Pi \alpha: \U.\ \Pi x, y: \alpha .\ \Pi p, q: x = y.\ p =  q
\end{align*}
{\bf \underline{$\Sigma$-Type}}
\begin{align*}
  \begin{aligned}
    \infer[\Sigma]{\Gamma \vdash \Sigma x:\sigma. \tau : \U}{%
      \Gamma \vdash \sigma : \U
           &
      \Gamma, x: \sigma \vdash \tau : \U
    }
    \qquad
    &
        \infer[\tup{-,-}]{\Gamma \vdash \tup{a,b} : \Sigma x: \sigma. \tau}{%
      \Gamma \vdash a : \sigma \quad \Gamma \vdash b : \tau[x := a]
    }
    \\\\
    \infer[\pr1]{\Gamma \vdash \pr1\ p : \sigma}{%
      \Gamma \vdash p : \Sigma x: \sigma. \tau
    }
    \qquad
           &
    \infer[\pr2]{\Gamma \vdash \pr2\ p : \tau[x := \pr1\ p]}{%
      \Gamma \vdash p : \Sigma x: \sigma. \tau
    }
  \end{aligned}
\end{align*}
The rules for $\Sigma$ also comes with equality rules:
\begin{eqnarray*}
  \pr1\ \tup{a, b} &\betar& a \\
  \pr2\ \tup{a, b} &\betar& b
\end{eqnarray*}
      
\end{definition}

The well-known encoding of inductive data types in polymorphic
$\lambda$-calculus extends immediately to $\lampt$.

\begin{example}\label{exa.natbool}
\begin{enumerate}
\item The natural numbers can be encoded by $\nat := \Pi
\alpha\oftype\star. \alpha\arr(\alpha\arr\alpha)\arr\alpha$, with zero
and successor:
\begin{eqnarray*}
\zero &:=& \lambda \alpha\oftype \star.\ \lambda x\oftype\alpha.\ \lambda
f\oftype\alpha\arr\alpha.\ x,\\
\suc &:=& \lambda n\oftype
\nat.\  \lambda \alpha\oftype \star.\ \lambda x\oftype\alpha.\ \lambda
f\oftype\alpha\arr\alpha.\ f\ (n\ \alpha\ x\ f).
\end{eqnarray*}
The induction principle reads
$$\ind_{\nat} := \Pi P: \nat\arr\star.\ P\ \zero \arr (\Pi y\oftype \nat .\ P\
y \arr P\ (\suc\ y))\arr \Pi x\oftype \nat.\ P\ x.$$
\item The type of booleans can be encoded by $\bool := \Pi
\alpha\oftype\star.\ \alpha\arr\alpha\arr\alpha$, with true and false:
\begin{eqnarray*}
\true &:=& \lambda \alpha\oftype \star.\ \lambda x,y\oftype\alpha.\ x\\
\false &:=& \lambda \alpha\oftype \star.\ \lambda x,y\oftype\alpha.\ y.
\end{eqnarray*}
The induction principle reads 
$$\ind_{\bool} := \Pi P\oftype
\bool\arr\star.\ P\ \true \arr P\ 
\false \arr \Pi x\oftype \bool.\ P\ x.$$
\item Equality is defined in $\lampt$ using Leibniz equality: for
$\sigma:\star$, $t,q:\sigma$
$$t=_{\sigma} q \:\: := \:\: \Pi P\oftype \sigma \arr \star.\ P\ t \arr
P\ q.$$
\end{enumerate}
\end{example}

One would like the induction principle to be provable in $\lampt$, but this is not the case, not for $\nat$ and not for $\bool$, and as a matter of fact also not for other definitions of the natural numbers or the booleans \cite{NoInduction}. There are also co-inductive data types that are definable in $\lampt$, for which one could hope for a co-induction principle. We show the example of streams, but before that, we recall the definable existential types, and we introduce some notation for them.

\begin{example}\label{exa.exists}
  For $\sigma$ a type, possibly containing the type variable $\alpha$,
  the existential type $\exists \alpha:\U.\ \sigma$ is defined as $\Pi
  \beta:\U.\ (\Pi \alpha:\U.\ \sigma \to\beta)\to\beta$.
\begin{itemize}
\item  For $\tau:\U$ and $f: \sigma [\alpha :=\tau]$, we have a term
  $\pack\ \tau\ f:= \lambda \beta:\U.\ \lambda k: \Pi
  \alpha:\U.\ \sigma \to\beta.\ k\ \tau\ f$, so $\pack : \Pi
  \beta:\U.\ \sigma [\alpha :=\beta] \to \exists \alpha:\U.\ \sigma$ is
  the term encoding $\exists$-introduction.
\item  For elimination, we have, for $\tau$ a type with $\alpha
  \notin\FV(\tau)$, $g: \Pi \alpha:\U.\ \sigma \to \tau$, and $t:
  \exists \alpha:\U.\ \sigma$, the term $\RecExists\ \tau\ g\ t :=
  t\ \tau\ g : \tau$. So $\RecExists : \Pi \beta:\U.\ (\Pi\alpha
  :\U.\ \sigma \to\beta) \to (\exists \alpha:\U.\ \sigma) \to \beta$
  is the $\exists$-elimination term, defined by $\RecExists := \lambda
  \beta:\U.\ \lambda g: \Pi\alpha :\U.\ \sigma \to\beta.\ \lambda t:
  \exists \alpha:\U.\ \sigma.\ t\ \beta\ g$.  
\end{itemize}
\end{example}

In $\lampt$ we also have the well-known product data type $\sigma
\times \tau$ that we will not spell out here. Pairing and projection are
definable: for $a :\sigma$ and $b:\tau$ we have $\tup{a,b} :\sigma
\times\tau$ and we have $\pr1 :\sigma\times \tau \to \sigma$ and $\pr2
:\sigma\times \tau \to \tau$ with $\pr1\ \tup{a,b} \redb a$ and
$\pr2\ \tup{a,b} \redb b$. For functions on product types, we allow
ourselves some syntactic sugar via ``pattern abstractions'', writing
$\lambda \tup{x,y}:\sigma\times \tau. \ldots$ for $\lambda
z:\sigma\times\tau.\ldots$ where in the second case we have to refer
to $\pi_1\ z$ and $\pi_2\ z$ and in the first we can just refer to $x$
and $y$.

\begin{example}\label{exa.streams}
The {\em stream type} over type $\sigma$, denoted by $\Stream$, is defined as $\exists \alpha:\U. \alpha \times (\alpha \to \sigma) \times (\alpha \to \alpha)$\footnote{If we spell it out fully, unfolding the definition of existential types as given in Example \ref{exa.exists}, we have $\Stream = \Pi \beta:\U.\ (\Pi \alpha:\U. \alpha \to (\alpha \to \sigma) \to (\alpha \to \alpha) \to \beta) \to \beta$}.
  The destructor $\hd(s: \Stream)$ is defined as $\RecExists\ \sigma\ (\lambda \alpha. \lambda \tup{x, h, t}.\ h\ x)\ s$,
  and the destructor $\tl(s: \Stream)$ is \[ \RecExists\ \Stream\ (\lambda \alpha. \lambda \tup{x, h, t}.\ \pack\ \alpha\ \tup{t\ x, h,t})\ s \ : \ \Stream. \]
  Finally, we define the {\em corecursor} $\CoRecStream$ as follows.
  \[
    \CoRecStream := \lambda \alpha : \U.\ \lambda h : \alpha \to \sigma .\ \lambda t : \alpha \to \alpha .\ \lambda x : \alpha .\ \pack\  \alpha\  \tup{x, h, t}.
  \]
The corecursor and the destructors (the head and tail functions) compute as expected: for $h:\tau\to\sigma$, $t:\tau\to\tau$ and $x:\tau$ we have
\begin{eqnarray*}
  \hd\  (\CoRecStream\  \tau\  h\  t\  x)  \betar h\  x  &\qquad &
  \tl\  (\CoRecStream\  \tau\  h\  t\  x)  \betar \CoRecStream\  \tau\  h\  t\  (t\  x)
\end{eqnarray*}
There are various (equivalent) ways to make this stream type a final co-algebra, as discussed in \cite{BronsveldGvdW}. Here are two.
\begin{enumerate}
\item The stream data type satisfies the {\em co-induction principle\/} if bisimilarity of streams implies their equality, more precisely, we have a term of type
  $$\Pi s,t : \Stream.\ s \bisim t \rightarrow s = t,$$
  where $s \bisim t$ is defined as $\exists R : \Stream \arr\Stream\arr \U.\ (R\mbox{ is a bisimulation relation and } R\ s\  t)$, where $R$ is a {\em bisimilation} if it is symmetric and for all $x, y : \Stream$, if $R\ x\ y$ then
   $\hd\  x = \hd\  y$ and $R\  (\tl\  x)\  (\tl\  y)$.
\item The corecursor is unique: if, for $h:\tau \to \sigma$ and $t:\tau \to\tau$, we have a term $g: \tau \to\Stream$ with
  \begin{eqnarray*}
  \hd\  (g\  x)  = h\  x  &\qquad &
  \tl\  (g\  x)  = g\  (t\  x),
\end{eqnarray*}
  then $g = \CoRecStream\  \tau\  h\  t$
\end{enumerate}
In Theorem \ref{thm.nocoind} we will show that these definable streams are not a final co-algebra, so these two principles do not hold. 
\end{example}

\subsection{Quotients}\label{sec.quotients}
To turn the stream type of Example \ref{exa.streams} into a final
co-algebra of streams, as we have done in \cite{BronsveldGvdW}, we
need to quotient it by the bisimulation relation. So we now introduce
quotients.

\begin{definition}\label{def.polymquotients}
  Given $\sigma, \tau :\U$ and $R: \sigma \to \sigma \to \U$, we say that a function $f : \sigma \to \tau$ {\em respects} $R$ if for all $x, y : \sigma $ such that $R \ x \ y$, we have $f \ x = f \ y$, and we write $\Resp\  f\  R$ for the type $\Pi x, y : \sigma .\ R \  x \  y \to f \  x = f \  y$.
  We define the {\em quotient type}  $\quot\  \sigma \  R$ to be
  $$\quot\  \sigma \  R := \Pi C: \U.\ \Pi f: \sigma  \to C.\ \Resp\  f\  R \to C.$$
The {\em class function} $\cls : \sigma  \to \quot\ \sigma \ R$ is defined to be
  \[ \cls := \lambda d : \sigma .\ \lambda C : \U.\ \lambda f: \sigma  \to C.\ \lambda H : \Resp\  f\  R.\ f\  d.\]
Note that the hypothesis $H$ doesn't occur in the body $f\ d$; it is not needed to define $\cls$.
  We write $\sigma  / R := \quot\ \sigma \ R$.
We can lift a function $f: \sigma  \to E$ that respects $R$ to a function $\widehat{f}: \sigma  / R \to E$. This lifting is done by the {\em recursor for quotient types}, $\RecQuot$, which is
defined as follows.
  \[
    \RecQuot := \lambda C : \U .\ \lambda f : \sigma  \to C .\ \lambda H : \Resp\  f\  R .\ \lambda q : \sigma  / R .\ q \  C\  f\  H
  \]
  We write $\widehat{f} := \lambda q.\ \  \RecQuot\  C\  f\  H\  q$ if $C$ and $H$ are clear from the context.
The lifted function satisfies $\widehat{f} \circ \cls = f$.
\end{definition}


In \cite{BronsveldGvdW}, it is shown that, if $R\ x\ y$, then
$\cls\ x = \cls\ y$. The proof uses function extensionality. Then it is also shown how to obtain a strong form of quotients, by
using $\Sigma$-types to define the sub-type of $\sigma /R$ of terms $q$
that satisfy a $\LimQuot$ predicate: basically we want to define the
subset of $\sigma /R$ for which $\widehat{f}$ is the {\em unique\/} function for which
$\widehat{f} \circ \cls = f$.
This also raises the question whether
the definable quotient type $\sigma /R$ is not already a strong quotient
type itself in $\lampt$? Or if there isn't a different, ``smarter''
representation of quotient types in $\lampt$ that is strong? We will
prove these questions in the negative by first specifying what we would
expect from a quotient type.

\begin{definition}\label{def.quotient}
  A {\em parametric quotient type\/} is given by
  \begin{enumerate}
  \item for $\sigma :\U$ and $R:\sigma \to \sigma \to \U$, a type $\sigma /R$, that we call the {\em quotient type of $\sigma $ modulo $R$},
  \item a term $\cls : \Pi \alpha:\U.\ \Pi R:\alpha\to \alpha\to \U.\ \alpha \to \alpha/R$,
  \item such that the following type is inhabited (for every $\sigma:\U$): $\Pi x,y:\sigma .\ R\ x\ y \to \cls\ \sigma \ R\  x =_{\sigma /R} \cls\ \sigma\ R\  y$;
  \item a term construction $\RecQuot$ satisfying
    $$\begin{prooftree}
    \tau:\U \quad f : \sigma  \to \tau \quad p: \Pi x,y:\sigma .\ R\ x\ y \to f\ x =_{\tau} f\  y  
    \justifies
    \RecQuot\ p\ f : \sigma /R \to \tau
  \end{prooftree}$$
    where we abbreviate $\RecQuot\ p\ f$ to $\widehat{f}$
  \item such that $\widehat{f} (\cls\ x) \betar f\ x$ for every $x:\sigma $.
  \end{enumerate}
  We call $\sigma /R$ a {\em strong parametric quotient type\/} if we have, in addition to the above
    \begin{enumerate}
    \item[6.] a term construction $\IndQuot$ satisfying
    $$\begin{prooftree}
    P:\sigma /R \to\U \quad t: \Pi x:\sigma .\ P\ (\cls\ x)  
    \justifies
    \IndQuot\ t : \Pi y:\sigma.\ R. P\ y  
  \end{prooftree}$$
    \end{enumerate}
\end{definition}

\weg{
  \begin{figure}\label{quot-eta}
  \centering
  $\begin{array}{lcr}
      \begin{tikzcd}
        D \arrow[d, "f"'] \arrow[r, "\clsr"] & D /^* R \arrow[ld, "\exists! \widehat{f}", dashed] \arrow[ld, "=", phantom, shift right=4] \\
        C                                   &
      \end{tikzcd}
       & \qquad &
      \begin{tikzcd}
        D \arrow[d, "g"'] \arrow[r, "\clsr"] \arrow[dd, "g'"', bend right=49] & D /^* R \arrow[ld, "\widehat{g}"] \arrow[ldd, "\widehat{g'}", bend left] \\
        X \arrow[d, "f"']                                                    &                                                                        \\
        Y                                                                    &
      \end{tikzcd}
    \end{array}$
\end{figure}
}

In the definition of the notion of quotient type, we do not require that $R$ is an equivalence
  relation. To define $\widehat{f}$ we require $\Resp\  f\  R$, i.e.\ $f$ should respect $R$, but of course $=_{\sigma /R}$ is an equivalence relation, so in fact, the type $\sigma /R$ in Definition \ref{def.quotient} is the
  quotient of $\sigma $ by the smallest equivalence relation containing $R$ \footnote{A similar idea was used by Hofmann~\cite{Hofmann95} to define quotients.}.

We could expect more from a strong quotient then what we required in
Definition \ref{def.quotient}, and some of these properties are
derivable, sometimes under weak additional assumptions. For example,
$\widehat{\cls}\  y =_{\sigma/R} y$ (for $y:\sigma/R$) is derivable from $\IndQuot$
and the uniqueness of $\widehat{f}$ (the $\eta$-rule for quotients) is
derivable if we also have functional extensionality. Finally, one could desire the quotient to be {\em effective\/} by requiring  $\Pi x,y:\sigma .\ \cls\ \sigma \ R\  x =_{\sigma /R} \cls\ \sigma\ R\  y \to R\ x\ y$. 
In \cite{BronsveldGvdW}, we show how to construct parametric quotients in an extension of $\lampt+\funext$. We do not exploit this further here, as we can already show that the parametric quotients of Definition \ref{def.quotient} don't exist in $\lampt$, in Proposition \ref{thm.no-quotients}.

It should be noted that we do not prove a general statement like {\em quotients are not definable in $\lampt$\/} simply because it is false in general: for specific $\sigma$ and $R$, the quotient $\sigma/R$ does exist. We give the example of the trivial quotient where we take $R$ to be Leibniz equality.

\begin{example}\label{exa.trivquotient}
Taking $\sigma:\nat$ and $R\ x\ y := x=_{\nat} y$, we define $\nat/R
:= \nat$ and we see that it is a non-parametric quotient in the sense  of Definition
\ref{def.quotient}, where we omit the requirement that $\cls$ is a polymorphic function.

To be precise, we have $\cls : \nat \to \nat/R$ defined by $\cls := \lambda x:\nat.\ x$, which obviously satisfies
\begin{itemize}
\item $\Pi x,y :\nat.\ R\ x \ y \to \cls\ x =_{\nat}\cls\ y$;
\item for $f :\sigma \to \tau$ with $\Resp\ f\ R$, taking $\widehat{f} := f$, we have $\widehat{f}\circ \cls = f$;
\item if $\Pi x:\nat.\ P\ (\cls\ x) $, then $\Pi y:\nat/R.\ P\ y $;
\item $\Pi x,y :\nat.\  \cls\ x=_{\nat}\cls\ y \to R\ x \ y $,
\end{itemize}
so this example of a trivial quotient is also an effective quotient.
\end{example}

\section{Model construction for $\lampt$}
The model notion for $\lampt$ we use is the same as in \cite{NoInduction,Geuvers1996}. It can be
extended to a class of models for the Calculus of Constructions, which
is done in \cite{stegeu95}.
Our models are a kind of realizability semantics in the sense that types will be interpreted as sets of realizers, where the realizers are taken from a combinatory algebra. The terms will be interpreted as the realizers. To cater for dependent types, our models are built from {\em weakly extensional
  combinatory algebras} (\weca). A {\em combinatory algebra}
  (\ca\ for short) is a tuple 
$\cA = \langle \bA,\cdot, \bk, \bs\rangle$, with $\bA$ a set,
$\cdot$ a binary function from $\bA\times \bA$ to $\bA$ (as usual
denoted by infix notation),
$\bk,\bs\in\bA$ such that
$(\bk \cdot a)\cdot b = a$ and $((\bs \cdot a)\cdot b)\cdot c =
(a\cdot c)\cdot (b\cdot c)$. For $\cA$ a combinatory algebra, the
set of {\em terms over $\cA$}, $\cT(\cA)$, is defined by letting
$\cT(\cA)$ contain infinitely many variables $v_1, v_2, \ldots$ and
distinct elements $c_a$ for every $a\in\bA$, and letting $\cT(\cA)$
be closed under application (the operation $\cdot$). Given a
term $t$ and a valuation $\rho$, mapping
variables to elements of $\bA$, the {\em interpretation of $t$ in
  $\bA$ under $\rho$}, notation $\inteAr{t}$, is defined in the usual
way ($\inteAr{c_a} = a$,  $\inteAr{MN} = \inteAr{M}\cdot \inteAr{N}$,
  etcetera). A crucial property of \ca s is 
that they are {\em combinatory complete}: if $t[v] \in\cT(\cA)$
is a term with free variable $v$, then there is an element in
$\bA$, usually denoted by $\lamstar v .\ t[v]$, such that $\inteAr{(\lamstar v .\ t[v])} \cdot a = \inteArx{t[x]}{a}$ for
  all $\rho$ and $a\in \bA$.
A \ca\ is {\em weakly extensional\/} if $\inteArx{t_1}{a} =
\inteArx{t_2}{a}$ for all $a\in \bA$ implies $\inteAr{\lamstar
x. t_1} = \inteAr{\lamstar x. t_2}$.


The need for {\em weakly extensional\/} \ca s comes from the fact that
  to make sure that the conversion rule is sound, we want the following to hold in our $\lampt$-model.
$$M=_{\beta}N \implies \intyr{M} = \intyr{N} \mbox{ for all }\rho,$$
where $\intyr{-}$ interprets pseudo-terms  as elements of $\bA$, using a
 valuation $\rho$ for the free variables. The interpretation $\intyr{-}$ is
 close to $\inteAr{-}$, where we now also
interpret abstraction: under $\intyr{-}$, $\lambda$ is interpreted as
$\lamstar$. 

\begin{example}\label{exaweca}
\begin{enumerate}
\item A standard example of a \weca\ is $\bLambda$, consisting of the
  classes of open $\lambda$-terms modulo $\beta$-equality. So, $\bA$
  is just $\Lambda/\beta$ and $[M] = [N]$ iff $M=_{\beta} N$.  It is
  easily verified that this yields a \weca.
\item Given a set of constants $C$, we define the \weca\ $\bLambda(C)$
  as the equivalence classes of open $\lambda_C$-terms
  (i.e.\ lambda-terms over the constant set $C$) modulo $\beta
  c$-equality, where the $c$-equality rules says
$$c N =_c c.$$
for all $c\in C$ and $N\in \Lambda_C$.
\item Another example of a \weca\ is $\bone$, the {\em degenerate\/}
  \weca\ where $\bA= 1$, the one-element set. In this case $\bk =\bs$,
  which is usually not allowed in combinatory algebras, but note that
  we do allow it here.
\end{enumerate}
\end{example}

The types of $\lampt$ will be interpreted as subsets of $\bA$.

\begin{definition}\label{def.polyset} 
A {\em polyset structure\/} over the weakly extensional 
combinatory algebra $\cA$ is a collection $\cP\subseteq \wp(\bA)$ such
that
\begin{enumerate}
\item $\bA \in \cP$,
\item $\emptyset\in \cP$\footnote{In \cite{NoInduction,Geuvers1996,stegeu95} we don't require this, and we call models where $\emptyset\in \cP$ ``consistent'', but here we only want to talk about consistent models anyway.},
\item $\cP$ is closed
  under arbitrary intersection $\bigcap$, \item $\cP$ is closed under
  {\em dependent products}, i.e.\ if $X\in \cP$ and $F:X\rightarrow
  \cP$, then $\Prodt_{t\in X} F(t)\in\cP$, where $\Prodt_{t\in X}
  F(t)$ is defined as
$$\{ a \in \bA \ |\ \forall t\in X (a\cdot t \in F(t) )\}.$$
  \end{enumerate}
The elements of a polyset structure are called {\em polysets}. If
$F$ is the constant function with value $Y$, we write
$X\arr Y$ instead of $\Prodt_{t\in X} Y$.
\end{definition}

\begin{example}\label{exapolyset}
\begin{enumerate}
\item We obtain the {\em full polyset structure\/} over the \weca\ $\cA$
if we take $\cP = \wp(\bA)$.
\item The {\em simple polyset structure\/} over the \weca\ $\cA$ is
obtained by taking $\cP = \{ \emptyset, \bA \}$. It is easily verified
that this is a polyset structure.
\item Given the \weca\ $\bLambda(C)$ as defined in Example \ref{exaweca}
(so $C$ is a set of constants), we define the {\em
polyset structure generated from $C$}, $\genC$ by
$$\genC := \{ X\subseteq \bLambda(C) \mid X=\emptyset \vee C \subseteq X
\}.$$ 
To show that $\genC$ is a polyset structure, we need to verify that it is closed under dependent product, which is easy: if
$X\neq\emptyset$ and
$F(t) \neq\emptyset$ for all $t\in X$, then $C\subseteq \Prodt_{t\in X}
F(t)$, because for $c\in C$ and $t\in X$, $ct =_c c \in C \subseteq
F(t)$.

\item Given the \weca\ $\cA$ and a set $C\subseteq \bA$ such that $\forall
a,b \in \bA (a \cdot b \in C \implies a \in C$, we define the {\em
power polyset structure of C\/} by 
$$\cP := \{ X \subseteq \bA \ |\  X\subseteq C \vee X = \bA \}.$$ To check
that this is a polyset structure, one only has to verify that,
for $X\in\cP$ and $F:X\arr\cP$,
$\Prodt_{t\in X} F(t) \in \cP$. This follows
from an easy case distinction: $\forall t\in X (F(t) = \bA)$ or
$\exists t\in X(F(t) \subseteq C)$.\\ An interesting instance of a power
polyset structure is the one arising from $C=\mbox{HNF}$, the set of
$\lambda$-terms with a head-normal-form, in the \weca\ $\Lambda/\beta$.
\end{enumerate}
\end{example}

To interpret kinds we need a {\em predicative structure}.

\begin{definition}\label{def.predstruct} For $\cP$ a polyset structure, the {\em predicative
    structure over\/} $\cP$ is the collection of sets $\cN$ defined 
inductively by
\begin{enumerate}
\item $\cP\in \cN$,
\item If $X\in \cP$ and $\forall t\in X(F(t) \in \cN$, then 
	$\Prodk_{t\in X} F(t) \in \cN$.
\end{enumerate}
If $F$ is a constant function with value $\cP$, we write $X\arr \cP$
in stead of $\Prodk_{t\in X} \cP$. 
\end{definition}

\begin{definition}\label{def.model}
If $\cA$ is a combinatory algebra, $\cP$ a polyset structure over
$\cA$, then we call $\langle \cA,\cP\rangle$ a {\em polyset model}.
\end{definition}

The predicative structure $\cN$ over a polyset structure $\cP$ could also be given as an additional parameter of a polyset model, if we take Definition \ref{def.predstruct} as a closure condition for $\cN$. However, we don't need that generality here. It is intended to
give a domain of interpretation for the kinds. For example, if the
type $\sigma$ is interpreted as the polyset $X$, then the kind $\sigma
\arr \sigma \arr \star$ is interpreted as $\Prodk_{t\in X}\Prodk_{q\in X}
\cP$, for which we usually write $X\arr X \arr \cP$.

We now define three interpretation functions, one for kinds, $\cV{-}$,
that maps kinds to elements of $\cN$, one for constructors (and
types), $\inte{-}$, that maps constructors to elements of $\bigcup \cN$ (and
types to elements of $\cP$, which is a subset of $\bigcup \cN$) and one for
terms, $\inty{-}$, that maps terms to elements of the combinatory
algebra $\cA$. All these interpretations are parameterized by
{\em valuations}, assigning values to the free variables (declared in
the context).

Let in the following $\cS= \langle \cA,\cP\rangle$ be a
polyset model: $\cA= \langle \bA,\cdot, \bk, \bs\rangle$ is
a combinatory algebra, $\cP$ is a polyset
structure over $\cA$ and $\cN$ is the predicative structure over the
polyset structure $\cP$.

\begin{definition} A {\em constructor variable valuation\/} is a map
  $\xi$ from $\Vc$ to $\bigcup \cN$. An {\em term variable
    valuation\/} is a map $\rho$ from $\Vo$ to $\bA$.
\end{definition}

\begin{definition} For $\rho$ an term variable
    valuation, we define the map $\intyr{-}^{\cS}$ from the set of
    terms to $\bA$ as 
    follows. (We leave the model $\cS$ implicit.)
    \begin{eqnarray*}
      \intyr{x} & :=  & \rho (x),\\
      \intyr{t q} & := & \intyr{t}\cdot\intyr{q}, \mbox{ if }q\mbox{
        is an term},\\
      \intyr{t Q} & := & \intyr{t}, \mbox{ if }Q\mbox{ is a
        constructor},\\
      \intyr{\lambda x\oftype \sigma . t} & := &\lamstar
      v. \intyrc{t}{\rho(x:=v)}, \mbox{ if }\sigma\mbox{ is a type},\\
      \intyr{\lambda \alpha\oftype A. t} & := & \intyr{t}, \mbox{ if
        }A\mbox{ is a kind}.
    \end{eqnarray*}
\end{definition}

\begin{definition} For $\rho$ an term variable
    valuation and $\xi$ a constructor variable
    valuation, we define the maps $\cVxr{-}^{\cS}$ and $\intexr{-}^{\cS}$
    respectively from kinds to $\cN$ and from constructors to $\bigcup
    \cN$ as follows. (We leave the model $\cS$ implicit.)
\begin{eqnarray*}
\cVxr{\star} &: = &\cP,\\
\cVxr{\Pi x \oftype \sigma . B} & := &
\Prodk_{t\in\intexr{\sigma}} \cVxrc{B}{\xi\rho(x:= t)},\\
\intexr{\alpha } & := & \xi (\alpha ),\\
\intexr{\Pi \alpha \oftype A . \tau} & := &
\bigcap _{a\in \cVxr{A}} \intexrc{\tau}{\xi(\alpha:=a)\rho}, \mbox{ if
  } A \mbox{ is a kind},\\
\intexr{\Pi x \oftype \sigma . \tau} & := &
\Prodt _{t\in\intexr{\sigma}}\intexrc{\tau}{\xi\rho(x:=t)}, \mbox{ if
  } \sigma \mbox{ is a type},\\
\intexr{P t} & := & \intexr{P}(\intyr{t}),\\
\intexr{\lambda x \oftype \sigma . P} & := & \llam t\in \intexr{\sigma}
.\intexrc{P}{\xi\rho(x:=t)}.
\end{eqnarray*}
\end{definition}



\begin{definition} For $\Gamma$ a $\lampt$-context, $\rho$ an term variable
    valuation and $\xi$ a constructor variable
    valuation, we say that $\xi,\rho$ {\em models\/} $\Gamma$,
    notation $\xi,\rho\models \Gamma$, if
    for all $x\in\Vo$ and $\alpha \in \Vc$,
      $x:\sigma \in \Gamma \Rightarrow  \rho(x)\in\intexr{\sigma}$ and
      $\alpha:A \in \Gamma \Rightarrow  \xi(\alpha)\in\cVxr{A}$.
\end{definition}

It is (implicit) in the definition that $\xi\rho\models \Gamma$ only
if for all declarations $x\oftype\sigma \in \Gamma$, $\intexr{\sigma}$ is
defined (and similarly for $\alpha \oftype A \in \Gamma$).

\begin{definition} The notion of {\em truth in a
  polyset model}, notation $\models^{\cS}$ and of {\em truth},
notation $\models$ are defined as follows. For
$\Gamma$ a context, $t$ an term, $\sigma$ a
  type, $P$ a constructor and $A$ a kind of $\lampt$,
  \begin{eqnarray*}
    \Gamma \models^{\cS} t:\sigma &\mbox{if}& \forall
    \xi,\rho[\xi,\rho\models \Gamma \Rightarrow
    \intyr{t}\in\intexr{\sigma}],\\
    \Gamma \models^{\cS} P:A &\mbox{if}& \forall
    \xi,\rho[\xi,\rho\models \Gamma \Rightarrow
    \intexr{P}\in\cVxr{A}].
  \end{eqnarray*}
Quantifying over the class of all polyset models, we define, for
$M$ an term or a constructor of $\lampt$,
$$    \Gamma \models M:T \;\mbox{ if }\; \Gamma \models^{\cS} M:T
    \mbox{ for all }\lampt\mbox{-models }\cS.$$
\end{definition}

Soundness states that if a judgment $\Gamma\vdash M:T$ is derivable,
 then it is true in all models. It is proved, by
 induction on the derivation in $\lampt$. As a matter of fact, Soundness also implies that $\cVxr{A}$ and $\intexr{P}$ 
are well-defined.

\begin{theorem}[Soundness] \label{thmsound}For $\Gamma$ a context, 
$M$ an term or a
  constructor and $T$ a type or a kind of $\lampt$,
$$    \Gamma \vdash M:T  \Rightarrow  \Gamma \models M:T.$$
\end{theorem}

\begin{example}\label{examodels}
Let $\cA$ be a \weca.
\begin{enumerate}
\item The {\em full polyset model\/} over $\cA$ is $\cS=
\langle \cA,\cP\rangle$, where $\cP$ is the full polyset structure
over $\cA$ (as defined in Example \ref{exapolyset} (1)). 
\item The {\em simple polyset model\/} over $\cA$ is $\cS=
\langle \cA,\cP\rangle$, where $\cP$ is the simple polyset structure
over $\cA$. (As defined in Example \ref{exapolyset} (2), so $\cP = \{ \emptyset, \bA \}$.)
\item The simple polyset model over the degenerate $\cA$ is also
called the {\em proof-irrelevance model\/} or {\em PI-model} for $\lampt$.
\item For $C$ a set of constants, the {\em
polyset model generated from $C$\/} is defined by $\cS= \langle
\bLambda(C),\genC\rangle$, where $\genC$ is the polyset structure
generated from $C$. See Example \ref{exapolyset} (3) and Example \ref{exaweca} (2).
\end{enumerate}
\end{example}

\section{Non-derivability results in $\lampt$}\label{sec.nonderiv}
We will not discuss generalities of our polyset models (see
\cite{NoInduction} for that), but instead we will show that certain
type constructions are not definable in $\lampt$ by constructing a concrete
counter-model.

In logic, validity of a formula $\phi$ means that the
interpretation of $\phi$ is true in the model. In type theory, we call
a type {\em valid\/} if its interpretation is nonempty, i.e.\ the type
is inhabited.

\begin{definition} For $\cS$ a polyset model, $\Gamma$ a context,
  $\sigma$ a type in $\Gamma$ and $\xi,\rho$ valuations such that
  $\xi,\rho\models \Gamma$, we say that {\em $\sigma$ is valid in $\cS$
  under $\xi,\rho$}, notation $\cS, \xi,\rho \Models^{\lampt}\sigma$, if
$$ \intexr{\sigma}^{\cS} \neq \emptyset.$$
We will often just write $\Models^{\lampt}\sigma$, in particular if the model $\cS$ is clear from the context and the type $\sigma$ is closed. 
\end{definition}

So we will sometimes just say that ``$\sigma$ holds'' when we mean
that the type $\sigma$ is inhabited in the type theory, and similarly
we will say that ``$\sigma$ holds in the model'' when we mean that the
interpretation of $\sigma$ is non-empty in the model.

We recap one important
Lemma about polyset models that we will use further on.

\begin{lemma} (Basically Lemma 4 of \cite{NoInduction}.) \label{lem.eq-in-model}
Let $\sigma$ be a type with terms $t, q :\sigma$, possibly in some
context.
\begin{itemize}
    \item For every polyset model $\cS$ we have
$$\cS,\xi,\rho \Models t=_{\sigma} q \quad\Longleftrightarrow\quad \intyr{t}= \intyr{q},$$
 where the equality on the right is the equality in the \weca\ of the model.
\item $\Gamma \vdash M : t=_{\sigma} q$ implies that for every polyset model $\cS$ and every $\xi,\rho$, if
$\xi,\rho \models \Gamma$, then $\intyr{t}= \intyr{q}$.
\end{itemize}
\end{lemma}

\begin{proof}
  The first is Lemma 4 of \cite{NoInduction}. The second follows immediately from that.
\end{proof}

\weg{
  To prove the non-definability of quotients in $\lampt$, we
consider a specific type $\sigma $ and relation $R$, and then we construct a
polyset model $\cS$ in which $\IndQuot$ is not valid (for
whatever quotient type of $\sigma $ modulo $R$ as in Definition
\ref{def.quotient}).
}

\begin{theorem}\label{thm.nocoind}
The definable streams in $\lampt$ (as given in Example \ref{exa.streams}) do not have a coinduction principle.
\end{theorem}

\begin{proof} We consider streams over $\bool$ and we define two streams $s_1, s_2 : \StreamB$ such that $s_1 \bisim s_2$ but not $s_1=_{\StreamB} s_2$. Define
  \begin{eqnarray*}
    s_1 &:=& \pack\  \bool\ \tup{\true, \neg, \neg}\\
    s_2 &:=& \pack\  \bool\ \tup{\false,\sI, \neg},
  \end{eqnarray*}
  where $\sI$ is the identity on $\bool$ and $\true$ and $\false$ are
  as in Example \ref{exa.natbool} and $\neg := \lambda
  b:\bool.\ b\ \bool\ \false\ \true$. Note that $\hd\ s_1
  = \neg\ \true =\false = \hd\ s_2$ and similarly $\hd(\tl\ s_1) = \neg\ (\neg\ \true) =\true$ and
  $\hd(\tl\ s_2) = \sI\ (\neg\ \false) = \true$: all finite observations from $s_1$ and $s_2$ are the same.

  We show that (1) $s_1 \bisim s_2$ but (2) not $s_1=_{\StreamB} s_2$.

  For (1): Define, for $x,y:\StreamB$, $R\ x\ y := (x =s_1\wedge
  y=s_2)\vee(y =s_1\wedge x=s_2)\vee(x =\tl\ s_1\wedge
  y=\tl\ s_2)\vee(y =\tl\ s_1\wedge x=\tl\ s_2)$. Then $R$ is
  symmetric and if $R\ x\ y$, then $\hd\ x=\hd\ y$, as we have already observed. Moreover, if $R\ x\ y$, then $R\ (\tl\ x)\ (\tl\ y)$ because $\tl\ (\tl,s_1) = s_1$ and $\tl\ (\tl,s_2) = s_2$.

  For (2): Suppose $s_1 =_{\StreamB} s_2$. Then $\inty{s_1}= \inty{s_1}$ in
  any model (Lemma \ref{lem.eq-in-model}), so if we consider the \weca\ of
  untyped $\lambda$-terms modulo $\beta\eta$, we see that
  $\inty{s_1}=_{\beta\eta} \inty{s_1}$. Now, if we fully spell out
  $\pack$, we have $\pack\  \sigma\ \tup{a, h, t} = \lambda
  Y:\U.\lambda k: \Pi X:\U. X \to (X \to \bool) \to (X \to X) \to
  Y. k\ \sigma\ a\ h\ t$, so we find that, in $\lambda$-calculus,
  $\inty{s_1}= \lambda k.\ k\ \bK\ N\ N = \lambda k.\ k\ \bKs\ \bI\ N =
  \inty{s_2}$\footnote{$\inty{\true} = \lambda x\ y.\ x = \bK$ and $\inty{\false} =\lambda x\ y.\ y = \bKs$}, where $N = \inty{\neg} = \lambda b.\ b
  \ \bKs\ \bK$. But $\lambda k. k\ \bK\ N\ N$ and $\lambda
  k. k\ \bKs\ \bI\ N$ are two distinct normal forms, so they are not
  $\beta\eta$-equal. Contradiction, so we conclude that not $s_1
  =_{\StreamB} s_2$\footnote{For (2), we do not really need the model, because we know that in $\lampt$, if $\vdash M: 
  s_1 =_{\StreamB} s_2$ (for some $M$),  then $s_1 =_{\beta} s_2$, which is not the case because $s_1$ and $s_2$ are different normal forms.}.
\end{proof}

\begin{theorem}\label{thm.no-quotients}
Parametric quotients (of Definition \ref{def.quotient}) are not definable in $\lampt$.
\end{theorem}

\begin{proof}
  We can show that the combination of the first five items in
  Definition \ref{def.quotient} cannot be realized in $\lampt$. Suppose we have a term
  $\cls : \Pi \alpha:\U. \Pi R:\alpha\to \alpha\to \U. \alpha \to
  \alpha/R$, such that, for all $\sigma:\U$, the type $\Pi x,y:\sigma
  .\ R\ x\ y \to \cls\ \sigma \ R\  x =_{\sigma /R} \cls\ \sigma\ R\ 
  y$ is inhabited. For $\alpha$ a type variable and taking $\lambda
  x,y:\alpha.\ \top$ for $R$ (where $\top := \Pi
  \beta:\U. \beta\to\beta$, the type 'True'), we see that there is a
  term $t$ with
$$\alpha :\U, x,y:\alpha \vdash t : \cls\ \alpha\ R\  x =_{\alpha/R} \cls\ \alpha\ R\  y.$$
  In a model built from the \weca\ $\bLambda$, taking $\intexr{\alpha} = \bLambda$, we see that for every $M, N\in \bLambda$, $\inty{\cls}\ M =_{\beta} \inty{\cls}\ N$. (This uses Lemma \ref{lem.eq-in-model}.) So $\inty{\cls}$ is a constant function in $\bLambda$. But this is not the case: if we take, e.g., $\sigma:=\bool$, $R\ x\ y := x=_{\bool} y$ and $f:= \sI_{\bool}$, then $\widehat{f} :\bool \to \bool$ with $\widehat{f}\ (\cls \ x) = x$ yields $\bK =\bKs$ in the model $\bLambda$, which is false. (If $\inty{\cls}$ is constant in the model $\bLambda$, then we have, in the model: $\true = \sI\ \true = \widehat{\sI}\ (\cls \ \true) = \widehat{\sI}\ (\cls \ \false) =  \sI\ \false = \false$.
\end{proof}

\section{Models of extensions of $\lampt$}\label{sec.extensions}
We now extend the polyset model of Definition \ref{def.model} to
incorporate identity types and $\Sigma$-types. We will show that the rules for identity
types and $\Sigma$-types of Definition \ref{def.extensions} are sound in the model, and
that $\UIP$ holds. We will also show a model in which $\funext$ does
not hold and neither does induction for natural numbers, indicating that
$\funext$ is crucial for the construction in \cite{encodings}.

\begin{definition} \label{def.lambda-id}
  We define $\bLambdaid$ by extending the untyped $\lambda$-calculus with two constants $\Jcomb$ and $\refl$ with the following reduction rules.
  $$\begin{array}{rclccrcl}
    \Jcomb \ c\ a\ b\ \refl &\redJ & c\  a &\qquad&     \refl\  q &\redrefl & \refl,\\
    \pr1\ \tup{a, b} &\redsig& a &&     \pr1\ \refl&\redrefl & \refl,\\
  \pr2\ \tup{a, b} &\redsig& b && \pr2\ \refl&\redrefl & \refl.
  \end{array}$$
\end{definition}

The rule for $\redJ$ represents the reduction rule for identity types that combines the $\Jcomb$ eliminator and the $\refl$ constructor, the rules for $\redsig$ represent the $\beta$-rules for $\Sigma$-types.
The rules for $\redrefl$ does not occur in a typed setting, as $\refl\ M$ and $\pr1\ \refl$ are never well-typed, but it is harmless, in the sense that it doesn't spoil the Church-Rosser property (see \cite{Mitschke77,KlopPhD}) and it allows us to define an interesting model.

The extension of the interpretation to $\Sigma$-types can be done by
extending the notion of polyset-structure (Definition
\ref{def.polyset}) to incorporate dependent sums in a straightforward
way. However, we are particularly interested in constructing a
counter-model, and this will be based on $\bLambdaid$ and the the
concrete polyset structure $\genidrefl$. So we will define the
concrete interpretation of $\Sigma$-types into $\genidrefl$.

\begin{definition}\label{def.sigmainterp}
  We take $\cP$ to be $\genidrefl$.
  For $X\in \cP$, we say that $F :X\rightarrow \cP$ is {\em proper\/} if $F(\refl)  =\emptyset \rightarrow \forall t\in X(F(t) =\emptyset)$.

  For  $X\in \cP$ and $F :X\rightarrow \cP$ proper, we define the {\em dependent sum\/} $\Sigma_{t\in X}. F(t)$ as follows.
$$\Sigma_{t\in X}. F(t) := \{ q \mid \pr1 \ q \in X \wedge \pr2\ q \in F(\pr1\ q)\}.$$ 
\end{definition}
Due to our choice of reduction rules for $\bLambdaid$ and the fact
that we limit to proper $F: X\rightarrow \cP$, we can see that
$\genidrefl$ is closed under dependent sums.

\begin{lemma}
    We take $\cP$ to be $\genidrefl$.
  For  $X\in \cP$ and $F :X\rightarrow \cP$ proper, the dependent sum $\Sigma_{t\in X}. F(t)$ is also in $\cP$.
\end{lemma}

\begin{proof}
  Let $X\in \cP$ and $F :X\rightarrow \cP$ proper. We need to show
  that $\Sigma_{t\in X}. F(t) =\emptyset$ or $\refl \in \Sigma_{t\in
    X}. F(t)$. Case $X=\emptyset$. Then $\Sigma_{t\in X}. F(t)
  =\emptyset$. Case $X\neq\emptyset$. Subcase $F(\refl) \neq
  \emptyset$. Then $\refl \in X$ and $\refl \in F(\refl)$, so $\refl
  \in \Sigma_{t\in X}. F(t)$. Subcase $F(\refl) = \emptyset$, then
  $\forall t\in X(F(t) =\emptyset)$, so $\Sigma_{t\in X}. F(t)
  =\emptyset$.
\end{proof}

The interpretation of $\Sigma$-types in a polyset structure is as follows.
$$\intexr{\Sigma x \oftype \sigma . \tau}  := 
\Sigma_{t\in\intexr{\sigma}}\intexrc{\tau}{\xi\rho(x:=t)}.$$

The interpretation of terms involving pairing and projection is
straightforward and so is the soundness of the derivation rules for
$\Sigma$-types, which we state without proof.

\begin{lemma}\label{lem.soundsig}
The polyset model $\genidrefl$ is sound for $\lampt^{\Sigma}$, that is, $\lampt$ extended with the
rules for $\Sigma$-types.
\end{lemma}

We now extend the interpretation to identity types. Note that
identity types are logically equivalent to Leibniz equality (we have
terms of type $\forall x,y :\sigma.\ x=_{\sigma} y \to x=^{\sigma} y$
and of type $\forall x,y :\sigma.\  x=^{\sigma} y \to x=_{\sigma} y$),
but there is no isomorphism between them, so we cannot just use the
interpretation of Leibniz equality in a polyset-model as the
interpretation for identity types.

\begin{definition}\label{def.interp-id}
  Let $\cP$ be a polyset-model. Let $\sigma$ be a type and $a,b:\sigma$. Given valuations $\xi,\rho$, we interpret the type $a=^{\sigma}b$ as follows. We write $a'$ for $\intyr{a}$, $b'$ for $\intyr{b}$ and $\sigma'$ for $\intexr{\sigma}$ to keep the notation readable.
  $$\intexr{a=^{\sigma}b} := \bigcap_{C \in \sigma' \to\sigma'\to \bA \to \cP} \{ q \mid \forall t \in \Pi_{x\in\sigma'}. C(x,x,\refl)\ [ J\ t\ a'\ b'\ q \in C(a',b', q)]\}.$$
\end{definition}

Note that $\intexr{a=^{\sigma}b}$ need not exist. We know that $\cP$ is closed
under arbitrary intersections, but the sets $\{ q \mid \forall t \in
\Pi_{x\in\sigma'}. C(x,x,\refl)\ [ J\ t\ a'\ b'\ q \in C(a',b', q)]\}$
need not be in $\cP$.

\begin{lemma}\label{lem.interp-id}
  The interpretation of $a=^{\sigma}b$ in Definition \ref{def.interp-id} is well-defined for
  \begin{enumerate}
  \item the full-polyset model $\cP := \wp(\bLambdaid)$;
  \item the polyset model generated from $C:= \{\refl\}$, $\genidrefl$.
  \end{enumerate}
\end{lemma}

\begin{proof}
For the first, this is immediate, for the second, we distinguish cases: (1) If $a'=b'$,
then $\refl\in \intexr{a=^{\sigma}b}$.  (2) If $a'\neq b'$, we can choose
a $C$ such that $C(a',b',q) =\emptyset$ for all $q\in \bA$, while $C(t,t,q) =\bA$ for all $t\in\sigma'$ and all $q\in \bA$. Then $\intexr{a=^{\sigma}b} =\emptyset$.
\end{proof}

\begin{lemma}\label{lem.sound}
The rules for the identity type are sound in a polyset-model where the interpretation of $a=^{\sigma}b$ is well-defined.
\end{lemma}

\begin{proof}
  We need to check the rule for $\refl$ and the rule for $\Jcomb$.
  \begin{itemize}
  \item[$\refl$] We have to check that $\refl \in
    \intexr{a=^{\sigma}a}$ for any type $\sigma$ and $a:\sigma$.  This
    is the case, because, writing again $a'$ for $\intyr{a}$, we have
    $t\ a' \in C(a',a', \refl)$ and $J\ t\ a'\ a'\ \refl =t\ a'$, so
    $J\ t\ a'\ a'\ \refl \in C(a',a', \refl)$.
  \item[$\Jcomb$] We look at the hypotheses of the  $\Jcomb$-rule: Let $\Gamma$ be a context and within $\Gamma$ we have $\sigma:\U$, $a,b:\sigma$, $q:
    a=^{\sigma}b$ and $c : \Pi z: \sigma.\ \tau[x,y:=z,p := \refl]$. Furthermore we have  $\tau$ with $\Gamma, x: \sigma, y: \sigma, p: x
    =^{\sigma} y \vdash \tau : \U$. Let 
    $\xi,\rho\models \Gamma$. We have to show that $\intyr{\Jcomb
      (c,a,b,q)} \in \intexr{\tau[x:=a, y:=b, p:=q]}$.  As before we
    abbreviate $\intyr{a}$ to $a'$, $\intexr{\sigma}$ to $\sigma'$
    etc.\\
    Interpreting $\lambda x,y:\sigma.\ \lambda p:x =^{\sigma} y.\  \tau$ in the model 
    we get an element $C \in \Pi_{x,y\in\sigma'}.\intexr{x=^{\sigma}y} \to \cP$,
    and we need to show that $\intyr{\Jcomb (c,a,b,q)} \in
    C(a',b',q')$. We extend $C$ to $C_0 \in \sigma'\to\sigma'\to
    \bA \to \cP$ by defining $C_0(x,y,r) := \emptyset$ for $r\notin
    \intexr{x=^{\sigma}y}$.\\
    From $c' \in \Pi_{z\in \sigma'}.\ C(z,z,\refl)$ we get $c' \in \Pi_{z\in \sigma'}.\ C_0(z,z,\refl)$ and from $q'\in \intexr{a=^{\sigma}b}$ we conclude that $\Jcomb\ c'\ a'\ b'\ q' \in C_0(a',b',q') =C(a',b',q')$ and we are done. \qedhere
  \end{itemize}
\end{proof}

\begin{lemma}\label{lem.uip}
$\UIP$ holds in a polyset-model where the interpretation of $a=^{\sigma}b$ is well-defined.
\end{lemma}

\begin{proof}
  Let $p\in \intexr{a=^{\sigma}b}$ for $a,b:\sigma$. Take $C \in
  \sigma' \to\sigma'\to \bA \to \cP$ with $C(-,-,\refl) =\bA$ and
  $C(-,-,r) =\emptyset$ if $r\neq \refl$. Then there is a $t \in
  \Pi_{x\in\sigma'}.C(x,x,\refl)$ and for such a $t$ we have
  $\Jcomb\ t\ a'\ b'\ p \in C(a',b',p)$. This means that $C(a',b',p)\neq
  \emptyset$, so $p =\refl$ in $\bA$.
\end{proof}

In the following we restrict ourselves to the \weca\
$\bA:= \bLambdaid$.

\begin{lemma}\label{lem.funext}
  $\funext$ does not hold in the polyset model generated from $C:= \{\refl\}$, $\genidrefl$.
\end{lemma}

\begin{proof}
  Consider $f, g: \bool \to \bool$ defined by
  \begin{eqnarray*}
    f&:=& \lambda b: \bool.\ b\ \bool\ \true\ \false\\
    g&:=& \lambda b:\bool.\ b\ \bool\ \true\ \false\ \bool\ \true\ \false.
  \end{eqnarray*}
  We show that (1)
$\intexr{\Pi b:\bool. f\ x=^{\bool}g\ x} \neq \emptyset$ and (2)
$\intexr{ f=^{\bool\to\bool}g} = \emptyset$.

(1) The interpretation of $\bool$ is $\bool' := \cap_{x\in\cP} X\to
X\to X$, which is $\{\bK,\bKs,\refl, \lambda x.\ \refl, \lambda x,y.\ \refl\}$, where $\bK
= \lambda x\ y.\ x = \inty{\true}$ and $\bKs = \lambda x\ y.\ y =
\inty{\false}$. The interpretations of $f$  and $g$ in the model are $f' := \lambda b. b\ \bK\ \bKs$ and $g' := \lambda b. b\ \bK\ \bKs\ \bK\ \bKs$ 

We claim that we have for all $b\in\bool'$:
$$\refl \in \bigcap_{C \in \bool' \to\bool'\to \bA \to \cP} \{ q \mid
\forall t \in \Pi_{x\in\bool'}. C(x,x,\refl)\ [
  J\ t\ (f'\ b)\ (g'\ b)\ q \in C((f'\ b),(g'\ b), q)]\}.$$
This is because for $b =\bK$, $b=\bKs$, $b=\refl$, $b= \lambda
x.\ \refl$ and $b= \lambda
x,y.\ \refl$,  we have $f' b = g' b$ in $\bLambdaid$, so for any $C$,
$J\ t\ (f'\ b)\ (g'\ b)\ \refl = t\  (f'\ b) \in C((f'\ b),(g'\ b),
\refl)$. So, we have a term in $\Pi_{b\in\bool'}. \inte{f\ x=^{\bool}g\ x}$.

(2) If $q \in \intexr{ f=^{\bool\to\bool}g}$, then take a $D \in (\bool' \to\bool')\to (\bool' \to\bool') \to \bA \to \cP$ such that $D(h,h,\refl) = \bA$ for any $h\in \bool'\to\bool'$ and $D(f',g', -) =\emptyset$. Then we have a term\\ $t \in \Pi_{h\in \bool'\to\bool'}.D(h,h,\refl)$ and so $\Jcomb\ t\ f'\ g'\ q \in D(f',g',q)$. Contradiction, because $D(f',g',q) =\emptyset$. So $\intexr{ f=^{\bool\to\bool}g} = \emptyset$.

(NB. That $f'\neq g'$ in $\bLambdaid$ follows from the Church-Rosser property. From $f'= g'$ it follows that all terms in $\bLambdaid$ are equal: just consider $f'\ \bI\ \bK$ and $g'\ \bI\ \bK$.) 
\end{proof}

\begin{lemma}\label{lem.no-ind}
  Induction for the natural numbers does not hold in the polyset model generated from $C:= \{\refl\}$, $\genidrefl$.
\end{lemma}

\begin{proof}
  The proof is basically the same as in \cite{NoInduction}, where it
  is shown that for polyset-model $\cS$, we have: if $\cS \Models
  \ind_{N,0,S}$, then $\inte{N} = \{ S^n 0 \mid n\in \IN \}$. But in
  the present model $\inte{N}$ also contains $\refl$.
\end{proof}

We conclude with the following theorem, that collects together the
results of Lemmas \ref{lem.sound}, \ref{lem.uip} and \ref{lem.no-ind}.

\begin{theorem}\label{thm.nonderlampt}
  In $\lampt$ extended with identity types and \UIP, we cannot define a data type for natural numbers for which we can prove the induction principle.
\end{theorem}

\section{Conclusion and Discussion}
We have shown a number of {\em non-derivability results\/} in
polymorphic dependent type theory $\lampt$. We believe that the proofs
go through without change if we replace $\lampt$ with the pure
Calculus of Constructions \CC.

In Section \ref{sec.nonderiv} we have looked at quotient types and coinductive types, in particular the data type of streams.
It would be interesting to generalize these results in the following directions.
\begin{itemize}
\item To prove that there exists no stream type in $\lampt$ with a
  coinduction principle. Our Lemma \ref{thm.nocoind} just shows that
  the well-known definable streams don't have coinduction principle,
  but maybe one can give a smarter encoding of streams.
  \item Theorem \ref{thm.no-quotients} proves that parametric
    quotients don't exists. On the other hand, Example
    \ref{exa.trivquotient} gives a concrete example of a
    non-parametric quotient, but there the relation is really trivial:
    just the identity itself. Are there other non-parametric quotients
    (besides the trivial one)? Taking $R\ x\ y := \top$ we can define
    a $\cls$ function and $\widehat{f}$, but $\IndQuot$ is already
    unclear. In \cite{nuo} it has been studied which quotient types
    are definable and which ones are not. But this is in predicative
    type theory with inductive types, which has more expressive power.
  \item With functional extensionality, the definable quotient is
    parametric, as shown in \cite{BronsveldGvdW}. But to construct a
    strong quotient type, we add $\Sigma$-types, identity types and
    \UIP. Can one prove that there are no strong quotients in $\lampt
    +\funext$?
\end{itemize}

In Section \ref{sec.extensions} we have looked at extensions of $\lampt$ where inductive types (including quotient types) and coinductive types can be define with their proper induction and coinduction principles, as studied in \cite{encodings,BronsveldGvdW}. We have seen that if we add $\Sigma$-types and an identity type with \UIP, we still cannot have an encoding of natural numbers with an induction principle. We need to also add \funext. This theorem is not limited to natural numbers, but holds for any
algebraic data type, like booleans, lists, trees, ... Combined with
the results from {\cite{encodings}, this means that \funext\ is crucial for defining induction.

\bibliographystyle{eptcs}
\bibliography{references}
\end{document}

%% file: prooftree.tex
\newdimen\proofrulebreadth \proofrulebreadth=.05em
\newdimen\proofdotseparation \proofdotseparation=1.25ex
\newdimen\proofrulebaseline \proofrulebaseline=2ex
\newcount\proofdotnumber \proofdotnumber=3
\let\then\relax
\def\hfi{\hskip0pt plus.0001fil}
\mathchardef\squigto="3A3B
\newif\ifinsideprooftree\insideprooftreefalse
\newif\ifonleftofproofrule\onleftofproofrulefalse
\newif\ifproofdots\proofdotsfalse
\newif\ifdoubleproof\doubleprooffalse
\let\wereinproofbit\relax
\newdimen\shortenproofleft
\newdimen\shortenproofright
\newdimen\proofbelowshift
\newbox\proofabove
\newbox\proofbelow
\newbox\proofrulename
\def\shiftproofbelow{\let\next\relax\afterassignment\setshiftproofbelow\dimen0 
}
\def\shiftproofbelowneg{\def\next{\multiply\dimen0 by-1 }%
\afterassignment\setshiftproofbelow\dimen0 }
\def\setshiftproofbelow{\next\proofbelowshift=\dimen0 }
\def\setproofrulebreadth{\proofrulebreadth}

\def\prooftree{
\ifnum  \lastpenalty=1
\then   \unpenalty
\else   \onleftofproofrulefalse
\fi
\ifonleftofproofrule
\else   \ifinsideprooftree
        \then   \hskip.5em plus1fil
        \fi
\fi
\bgroup
\setbox\proofbelow=\hbox{}\setbox\proofrulename=\hbox{}%
\let\justifies\proofover\let\leadsto\proofoverdots\let\Justifies\proofoverdbl
\let\using\proofusing\let\[\prooftree
\ifinsideprooftree\let\]\endprooftree\fi
\proofdotsfalse\doubleprooffalse
\let\thickness\setproofrulebreadth
\let\shiftright\shiftproofbelow \let\shift\shiftproofbelow
\let\shiftleft\shiftproofbelowneg
\let\ifwasinsideprooftree\ifinsideprooftree
\insideprooftreetrue
\setbox\proofabove=\hbox\bgroup$\displaystyle 
\let\wereinproofbit\prooftree
\shortenproofleft=0pt \shortenproofright=0pt \proofbelowshift=0pt
\onleftofproofruletrue\penalty1
}

\def\eproofbit{
\ifx    \wereinproofbit\prooftree
\then   \ifcase \lastpenalty
        \then   \shortenproofright=0pt  
        \or     \unpenalty\hfil         
        \or     \unpenalty\unskip       
        \else   \shortenproofright=0pt  
        \fi
\fi
\global\dimen0=\shortenproofleft
\global\dimen1=\shortenproofright
\global\dimen2=\proofrulebreadth
\global\dimen3=\proofbelowshift
\global\dimen4=\proofdotseparation
\global\count255=\proofdotnumber
$\egroup  
\shortenproofleft=\dimen0
\shortenproofright=\dimen1
\proofrulebreadth=\dimen2
\proofbelowshift=\dimen3
\proofdotseparation=\dimen4
\proofdotnumber=\count255
}

\def\proofover{
\eproofbit 
\setbox\proofbelow=\hbox\bgroup 
\let\wereinproofbit\proofover
$\displaystyle
}%
\def\proofoverdbl{
\eproofbit 
\doubleprooftrue
\setbox\proofbelow=\hbox\bgroup 
\let\wereinproofbit\proofoverdbl
$\displaystyle
}%
\def\proofoverdots{
\eproofbit 
\proofdotstrue
\setbox\proofbelow=\hbox\bgroup 
\let\wereinproofbit\proofoverdots
$\displaystyle
}%
\def\proofusing{
\eproofbit 
\setbox\proofrulename=\hbox\bgroup 
\let\wereinproofbit\proofusing
\kern0.3em$
}

\def\endprooftree{
\eproofbit 
  \dimen5 =0pt
\dimen0=\wd\proofabove \advance\dimen0-\shortenproofleft
\advance\dimen0-\shortenproofright
\dimen1=.5\dimen0 \advance\dimen1-.5\wd\proofbelow
\dimen4=\dimen1
\advance\dimen1\proofbelowshift \advance\dimen4-\proofbelowshift
\ifdim  \dimen1<0pt
\then   \advance\shortenproofleft\dimen1
        \advance\dimen0-\dimen1
        \dimen1=0pt
        \ifdim  \shortenproofleft<0pt
        \then   \setbox\proofabove=\hbox{%
                        \kern-\shortenproofleft\unhbox\proofabove}%
                \shortenproofleft=0pt
        \fi
\fi
\ifdim  \dimen4<0pt
\then   \advance\shortenproofright\dimen4
        \advance\dimen0-\dimen4
        \dimen4=0pt
\fi
\ifdim  \shortenproofright<\wd\proofrulename
\then   \shortenproofright=\wd\proofrulename
\fi
\dimen2=\shortenproofleft \advance\dimen2 by\dimen1
\dimen3=\shortenproofright\advance\dimen3 by\dimen4
\ifproofdots
\then
        \dimen6=\shortenproofleft \advance\dimen6 .5\dimen0
        \setbox1=\vbox to\proofdotseparation{\vss\hbox{$\cdot$}\vss}%
        \setbox0=\hbox{%
                \advance\dimen6-.5\wd1
                \kern\dimen6
                $\vcenter to\proofdotnumber\proofdotseparation
                        {\leaders\box1\vfill}$%
                \unhbox\proofrulename}%
\else   \dimen6=\fontdimen22\the\textfont2 
        \dimen7=\dimen6
        \advance\dimen6by.5\proofrulebreadth
        \advance\dimen7by-.5\proofrulebreadth
        \setbox0=\hbox{%
                \kern\shortenproofleft
                \ifdoubleproof
                \then   \hbox to\dimen0{%
                        $\mathsurround0pt\mathord=\mkern-6mu%
                        \cleaders\hbox{$\mkern-2mu=\mkern-2mu$}\hfill
                        \mkern-6mu\mathord=$}%
                \else   \vrule height\dimen6 depth-\dimen7 width\dimen0
                \fi
                \unhbox\proofrulename}%
        \ht0=\dimen6 \dp0=-\dimen7
\fi
\let\doll\relax
\ifwasinsideprooftree
\then   \let\VBOX\vbox
\else   \ifmmode\else$\let\doll=$\fi
        \let\VBOX\vcenter
\fi
\VBOX   {\baselineskip\proofrulebaseline \lineskip.2ex
        \expandafter\lineskiplimit\ifproofdots0ex\else-0.6ex\fi
        \hbox   spread\dimen5   {\hfi\unhbox\proofabove\hfi}%
        \hbox{\box0}%
        \hbox   {\kern\dimen2 \box\proofbelow}}\doll%
\global\dimen2=\dimen2
\global\dimen3=\dimen3
\egroup 
\ifonleftofproofrule
\then   \shortenproofleft=\dimen2
\fi
\shortenproofright=\dimen3
\onleftofproofrulefalse
\ifinsideprooftree
\then   \hskip.5em plus 1fil \penalty2
\fi
}
